\documentclass[letterpaper,10 pt,conference]{IEEEtran}  % Comment this line out
\IEEEoverridecommandlockouts
\usepackage{cite}
\usepackage{multicol}
\usepackage{multirow}
\usepackage{amsthm,amsmath}
\usepackage{graphics}
\DeclareMathOperator*{\argmax}{argmax}
\usepackage{amssymb}
\graphicspath{{figures/}}
\usepackage{algorithm,algorithmic}
\usepackage{color}
\usepackage{mathtools}

\newtheorem{theorem}{Theorem}

\newtheorem*{remark}{Remark}
\usepackage{etoolbox}
\title{\LARGE \bf
Non-parametric Probabilistic Load Flow using Gaussian Process Learning}

\author{Parikshit Pareek, Chuan Wang and Hung D. Nguyen*
\thanks{$^\star$Corresponding Author}% <-this % stops a space
\thanks{Authors are with School of Electrical and Electronics Engineering, Nanyang Technological University, Singapore. \textit{pare0001,chuan002,hunghtd@ntu.edu.sg}}}%
\makeatletter
\patchcmd{\@maketitle}
  {\addvspace{0.5\baselineskip}\egroup}
  {\addvspace{-1\baselineskip}\egroup}
  {}
  {}
\makeatother

\begin{document}

\maketitle

%%%%%%%%%%%%%%%%%%%%%%%%%%%%%%%%%%%%%%%%%%%%%%%%%%%%%%%%%%%%%%%%%%%%%%%%%%%%%%%%
\begin{abstract}
In this work, we propose a non-parametric probabilistic load flow (NP-PLF) technique based on the Gaussian Process (GP) learning to understand the power system behavior under uncertainty for better operational decisions. The technique can provide ``\textit{semi-explicit}'' power flow solutions by implementing the learning and testing steps which map control variables to inputs. The proposed NP-PLF leverages upon GP upper confidence bound (GP-UCB) sampling algorithm. The salient features of this NP-PLF method are: i) applicable for power flow problem having power injection uncertainty with an unknown class of distribution; ii) providing probabilistic learning bound (PLB) which further provides control over the error and convergence; iii) capable of handling intermittent distributed generation as well as load uncertainties, and iv) applicable to both balanced and unbalanced power flow with different type and size of power systems. The simulation results performed on the IEEE 30-bus and IEEE 118-bus system show that the proposed method can learn the voltage function over the power injection subspace using a small number of training samples. Further, the testing with different input uncertainty distributions indicates that complete statistical information can be obtained for the probabilistic load flow problem with average percentage relative error of order $10^{-3}$\% on 50000 test points.
\end{abstract}


%%%%%%%%%%%%%%%%%%%%%%%%%%%%%%%%%%%%%%%%%%%%%%%%%%%%%%%%%%%%%%%%%%%%%%%%%%%%%%%%
\section{INTRODUCTION}
The decision-making process and control of power systems rely on the solution of optimal power flow (OPF) problem, solved by minimization of a variety of objectives such as real power generation cost. At the core of the OPF problem, there lies power balance equality constraints that govern the relationship between the voltages with the power injections, the so-called load flow problem. The non-convex load flow constraints pose obstacles in the decision-making process and in obtaining optimal setpoints, more so under uncertainty. Therefore, it is essential to understand such voltage-power relationship with the power input inside a uncertain space. With the increasing integration of renewable sources, such as wind power and photovoltaic (PV) energy, the generation parts of electrical systems may behave differently depending on the specific scenario \cite{da2018risk}. On the other hand, there is also uncertainty in the load demand side, such as the impact of electric vehicles. Therefore, an essential analysis tool is probabilistic load flow (PLF), which provides an understanding of equality constraint and does mapping between voltage variables and the power inputs. 
% Power flow plays an important role in analyzing power system operation status. Various deterministic power flow models are widely used to ensure system's operational reliability and security. When using a set of known data in a deterministic manner, the system's state parameters can be fully obtained. However, Thus, new power flow algorithms considering uncertainties in the power system are needed. 

The PLF was first proposed in \cite{borkowska1974probabilistic}, considering the impact of the input variables with known uncertainty distribution characteristics. The PLF approach can provide the probability density functions (PDFs) of the output values such as the power flow and nodal voltages. In literature, the methods that have been applied to solve the PLF problem can be divided primarily based on the type of tools employed, i.e., numerical methods and analytical methods \cite{allan1981evaluation}. One of the most widely used numerical methods is Monte-Carlo simulation (MCS), which enters randomly generated variables to calculate the corresponding output by a large number of repetitive power flow calculations \cite{hungzheng, constante2018data}. Although MCS can handle full nonlinear power flow equations, it requires a large number of iterations, which leads to a high computational cost and a long elapsed time. Further, many questions like how many points are sufficient and what can be the maximum error for any new point, cannot be answered using pure numerical approaches. 

The other class, analytical methods for PLF, is mainly a set of statistical approaches aiming to obtain the results of some special moments or the mean, variance, and PDFs for the state and output variables \cite{da2018risk}. Among them, convolution techniques follow mathematical assumptions in order to simplify the power flow problem \cite{hatziargyriou1993probabilistic}. But these methods do not apply to complex problems, such as AC power flow in unbalanced power distribution
systems \cite{nosratabadi2018nonparametric}.  Further, for preserving power flow equations non-linearity, a set of deliberated operating conditions are used in point estimate methods \cite{morales2007point}.
%\textcolor{red}{Point estimate methods calculate power flows at a number of deliberated operating conditions and can preserve the non-linearity of some systems \cite{su2005probabilistic,morales2007point}.}\
However, their accuracy is low in estimating high order moments of probability distributions, especially for complex systems with many inputs \cite{liu2018probabilistic}. The third analytical method is approximation expansions based on cumulants \cite{zhang2004probabilistic}. It is designed to obtain the cumulants of outputs from the cumulants of inputs through a simple mathematical process. %The distribution of the output is then recovered from the obtained output using approximation expansions. 
These analytical methods can indeed reduce the computational costs, however, they suffer main drawbacks like: 1) requirement of model simplifications and adjustment of parameters, essentially losing information of tails of PDFs; 2) need of linearizing power flow equations for a specific operating condition, etc. With the high penetration of distributed energy resources, their accuracy declines due to ignoring non-linearity \cite{liu2018probabilistic} of the power balance equation. Further, analytical methods are parametric and work on a fixed, specific type of uncertainty distribution only. Primarily, the numerical methods capture non-linearity but are slow while analytical methods need approximations leading to lower accuracy. 

%Nonparametric and parametic methods are PDFs of input random variables are usually known in power flow studies. However, since the power flow equations are nonlinear, PDFs of output random variables are unknown. This is more critical for AC power flow in unbalanced power distribution systems. 

In the view of the aforementioned limitations in understanding power flow behavior under uncertainty, this work presents a non-parametric probabilistic load flow (NP-PLF) method using Gaussian process regression. The growing renewable generation (especially PV) and electric vehicles lack historical data and efficient forecasting. Therefore, it is challenging to obtain accurate uncertainty distributions for these uncertainties. This makes it imperative to develop non-parametric methods which do not require input uncertainty distribution information.  The proposed method deals with two inherent difficulties faced in obtaining PLF solutions: i) the non-linearity of the power flow equation set, and ii) lack of statistical information about uncertainty and complexity in modeling random power injection PDFs as input. We develop probabilistic learning bound (PLB) using regret bounds of the so-called GP upper confidence bound (GP-UCB) sampling algorithm The PLB provides control over the desired accuracy and confidence level of learning. Here, the learning step of proposed NP-PLF works without any PDF of input uncertainty and testing step allows obtaining the state and output variables (e.g., nodal voltages) for any class and type input uncertainty distribution. This two-step method makes the proposed approach ``\textit{semi-explicit}'' in terms of the form of the power flow solution, while conventional numerical methods are implicit. By ``\textit{semi-explicit}'', we meant that we do not provide any fully analytical form of the voltage solutions, but the corresponding mean and error bound functions of power inputs expressed using the Gaussian process kernels such as those in \eqref{eq:gp}. 

The method can also serve as a multi-dimensional continuation power flow (CPF) \cite{ajjarapu1992continuation} to estimate the power-voltage curve with a given confidence level as shown in fig.\ref{fig:rg4_V25_gp}. Further, this proposed uncertainty propagation framework could help in understanding other nonlinear balance equation behavior to facilitate the decision-making process, thus lend itself to smart buildings, microgrids, and control of electric transportation systems such as EVs. 

The main contributions of the paper are as follows.
\begin{enumerate}
    \item A novel non-parametric probabilistic load flow (NP-PLF) method, which handles the uncertain power injections. The method is generic as it does not rely on the class of uncertainty distribution. It is fast and captures the non-linearity of power flow equations. The learning step provides a ``\textit{semi-explicit}'' form of power flow solutions, while the testing step can provide statistical information for the PLF solution. 
    \item Development of probabilistic learning bound (PLB) for the inverse power flow solution by GP learning and GP-UCB sampling. This PLB can serve as a convergence criterion for the learning stage algorithm and bounds the output solution within the specific range. Thus, PLB provides the confidence level for statistical information of output.
\end{enumerate}

\color{black}

\section{Background}\label{sec:gp}
In this section, we introduce GP regression and GP-UCB algorithm for sampling training points for GP learning bounded by regret bounds. Hereafter, $\mathbf{X}=\{\mathbf{x}^i \in \mathbb{R}^D\}$ represents matrix having $N$ training points and $\mathbf{y}=\{ y^i=y(\mathbf{x}^i) \in \mathbb{R} \}$ indicates $N$ observations over these training points for $i=1, \dots , N$.

\subsection{Gaussian Process Regression (GPR)}
The Gaussian process is the backbone of the Bayesian optimization paradigm. The analogy extension of the multivariate Gaussian function and interpretation of GP as a distribution over random functions \cite{williams2006gaussian} is intuitive and very useful as a non-parametric method, for modeling PDFs over functions. In power system, various forecasting applications for wind power \cite{lee2013short,yan2015hybrid}, solar power \cite{sheng2017short}, and electricity demand \cite{van2018probabilistic} have been developed based on GP. Other than these forecasting works, the idea of using GP to learn the dynamics and stability index behavior has been explored recently in \cite{zhai2019region,pareek2019probabilistic}.

%In power systems, the use of GP is limited to power and load forecasting applications. The GP has been applied to wind power forecasting \ cite {lee2013short, chen2013wind, yan2015hybrid}, solar power forecasting \ cite {sheng2017short} and power demand forecast \ cite {blum2013electricity, yang2018power, van2018probilistic}. In addition to these predictions, the idea of using GP to learn dynamic and stability indicator behavior has also been explored recently in \cite {zhai2019region, pareek2019probabilistic}.

The general non-parametric GP regression (GPR), with a training data set $\mathcal{D}=\{\mathbf{X},\, \mathbf{y} \}$, attempts to infer the latent function $f: \mathbb{R}^D \mapsto \mathbb{R}$ \cite{williams2006gaussian} with \textit{i.i.d.} noise $\varepsilon \sim \mathcal{N}(o,\sigma^2_\varepsilon)$:
\begin{align}\label{eq:reg}
         y(\mathbf{x}^i) =f(\mathbf{x}^i)+\varepsilon, \quad i=1 \dots N.
\end{align}

The prior over this latent function, with zero mean GP, is placed as $f(\mathbf{x}) \sim \mathcal{GP}\big ( 0,k(\mathbf{x},\mathbf{x}') \big )$. The kernel function $k(\mathbf{x},\mathbf{x}')$ incorporate our comprehension of unknown function into GP. We use squared exponential kernel for its smoothness property in this work. 
In the probabilistic load flow (PLF) problem, $\mathbf{x}$ can represent uncertain power injections such as those from intermittent renewable sources, and $ y(\mathbf{x})$ refers to the numerical voltage solution corresponding to a load injection point of $\mathbf{x}$.  Interested reader can look into \cite{williams2006gaussian} for details of GP fundamentals. 
   %A sample set  $\hat y_N=[\hat y(\mathbf{x}^{(1)}), \dots, \hat y(\mathbf{x}^{(N)})]^T$ at operating points $\mathcal{D}_N = \{\mathbf{x}^{(1)}, \dots ,\mathbf{x}^{(N)}\}$ with Gaussian noise $\varepsilon $, and the analytic formula set can be obtained for posterior distribution corresponding to (\ref{eq:reg}). 
 Following the extension of Bayesian rule \cite{williams2006gaussian}, analytical formula set for posterior  distribution for \eqref{eq:gp} are obtained as:
\begin{subequations}\label{eq:gp}
\begin{align}
    \mu(\mathbf{x}') & = k(\mathbf{x}',\mathbf{X})^T K^{-1}\mathbf{y}\\
    \sigma^2(\mathbf{x}') &= k(\mathbf{x}',\mathbf{x}')- k(\mathbf{x}', \mathbf{X})^T K^{-1} k(\mathbf{x}', \mathbf{X})
\end{align}
\end{subequations}

Here, $K = k(\mathbf{X},\mathbf{X})+\sigma^2 I$,  $\mu(\mathbf{x}')$ is the inferred mean, $k(\mathbf{x},\mathbf{x}')$ is the covariance, and inferred variance is given as $\sigma^2(\mathbf{x}')$ at input $\mathbf{x}'$ \cite{williams2006gaussian}. The random input vector $\mathbf{x'}$ can be drown from any type of uncertainty distribution.

\subsection{Gaussian Process Upper Confidence Bound (GP-UCB)}
Sampling schemes play an important role in learning latent function. This paper relies upon the GP-UCB, widely used in Bayesian optimization paradigm \cite{srinivas2012information}, for sampling of input vectors. The target is to obtain the mean $\mu(\mathbf{x})$ for function $f(\mathbf{x})$ with least standard deviation $\sigma(\mathbf{x})$ and probability at least $1-\delta$ with $\delta \in (0,1)$. A joint multi-objective function balancing between exploration and exploitation is opted for obtaining next input point $\mathbf{x}^{i}$. With $\beta_i$ taken independent of state vector and $\mathcal{S}$ being uncertain input subspace, the sampling strategy will be \cite{srinivas2012information} : 
\begin{align}\label{eq:gpucb}
    \mathbf{x}^{i} = \argmax_{\mathbf{x}^{i} \in \mathcal{S}}  \left \{ \mu_{i}(\mathbf{x}) + \sqrt{\beta_{i+1}}\,\,\sigma_{i}(\mathbf{x}) \right \}.
\end{align} 

 Intuitively, \eqref{eq:gpucb} means that next sampled input point will be the one where weighted sum, $\mu_{i}(\mathbf{x}) + \beta_{i+1}^{1/2}\sigma_{i}(\mathbf{x})$, of mean and variance is maximum.  
 %The $\mu_{i}(\mathbf{x})$ contributes to enlarging the level set of critical eigenvalue while $\sigma_{i}(\mathbf{x})$ helps in minimizing the uncertainty. 
 Interested readers can refer to \cite{srinivas2012information} for more detail of this sampling strategy and GP-UCB.

% \begin{theorem}\label{theo:UCB}
% Let $f(\mathbf{x})$ be the power flow solution function in state $\mathbf{x}$ with the noise $\epsilon$ bounded by $\sigma_n$. Then with $\delta\in(0,1)$, the following holds with  atleast the probability $1-\delta$: 
% \begin{align}\label{eq:rbound}
%     |f(\mathbf{x})-\mu_{N}(\mathbf{x})|\leq \sqrt{\beta_{N+1}} \, \sigma_{N}(\mathbf{x}), \quad \forall \mathbf{x}\in \mathcal{S}.
% \end{align}
% Here $\mathcal{X}$ is the sample space wherein the states $\mathbf{x}$ lie, and $m$ is the number of sampling points.  $\beta_N=2\|V\|_k^2+300\gamma_N\ln^3{(m/\delta)}$ defined in \cite{zhai2019region}.
% \end{theorem}
% \begin{proof}
% The results follow directly from Theorem 6 in \cite{srinivas2012information} and Theorem 1 in \cite{zhai2019region}. 
% \end{proof}

\section{Non-Parametric Probabilistic Load Flow (NP-PLF)}\label{sec:NP-PLF}

In this section, we present the proposed NP-PLF algorithm with the main result as a probabilistic guarantee bound on power flow solution. First, we present the generic power flow formulation and inverse power flow analogy for obtaining uncertain state variables corresponding to uncertain input. 

The relationship between apparent the power injection at $l^{th}$ and node voltages for an $n$-bus network is given as: 
\begin{align} \label{eq:pf}
      S_l  & = \sum^n_{m=1} |V_l||V_m|(\cos{\theta_{ml}}+j\sin{\theta_{ml}})(G_{ml}-jB_{ml}).
\end{align}

Here, $S_l$ is the complex apparent power injection in node $k$, $V_l=|V_l|\angle{\theta_l}$ is node voltage while $\theta_{ml}=\theta_m-\theta_l$ is angle difference between node $m$ and $l$. The $G_{kl}$ and $B_{kl}$ are conductance and susceptance of line connecting node $m$ and $l$. This complex non-linear power flow equation \eqref{eq:pf} is solved using iterative numerical methods such as Newton-Rapson in a deterministic manner for one set of input. There does not exist any explicit form for the relationship between the node voltage and power injections.  We define an input vector $\mathbf{x}= [S_1, \dots, S_l \dots, S_n]^T$ for $n$-bus network then the family of power flow equation \eqref{eq:pf} can be written in compact form as $ \mathbf{x} = h(\mathbf{y}) $ with $\mathbf{y}$ being collection of node voltage magnitude and angle. 

Now, we want to learn the inverse power flow function, mapping the voltage and power injection relationship using GPR as $|V_l|=f_l(\mathbf{x})$ where, $f_l: \mathbb{R}^n \mapsto \mathbb{R}$ is \textit{inverse power flow} function. In other words, we want to learn the relationship between voltage and power injection in uncertain subspace $\mathbf{x} \in \mathcal{S}$ of power injections. 

% The AC power flow equations are nonlinear in the nodal voltages, and can be expressed as: 
% \begin{align}\label{eq:lf}
%     \mathbf{x} & = h (\mathbf{y})
% \end{align}
% where the random input vector $\mathbf{y}$ includes active and reactive power injections at buses. The state random vector including node voltage magnitude $\mathbf{V}$ and angle $\mathbf{\theta}$ is indicated as $\mathbf{x}$.

% Here, $\mathbf{z}$ is output vector while the vector $\mathbf{x}$ is indicative to the input. Further, the input can have some uncertain variables while some certain or non-random. Similarly, any subset of the output can be the uncertain output which we want to learn using GP.

% Now, we want to learn a node voltage $V_i \in \mathbf{y}$ against  variations in the input of power injections. In other words, we want to learn the behavior of $V_i$ in uncertain input subspace $\mathcal{S}$ of the power injections. To learn the voltage solution as a function of the random power injections, we need to learn the inverse power flow function $\mathbf{y} = h (\mathbf{x})$ from \eqref{eq:lf}. This inverse function in the load flow context becomes
% \begin{align}\label{eq:ilf}
%  \mathbf{V} &= f(\mathbf{S}) 
% \end{align}
% where $f(\cdot) \equiv h^{-1}(\cdot)$. 

Here, we assume that power flow is solvable in the considered space of random power injections $\mathbf{x} \in \mathcal{S}$. Therefore, the inverse power flow is well-defined. Once the voltage solution is known, other network quantities such as branch flows or power loss can be calculated. We also assume that the network structure is unchanged for power flow studies without contingencies. 
 
% Or
% \begin{align}\label{eq:yfx}
%     \mathbf{y}=f(\mathbf{x}) 
% \end{align}

% We can say that $i$-th nodal voltage magnitude is a function of random variable $\mathbf{x}$ representing the random power injections, inside an input subspace $\mathcal{S}$ as $V_i=f(\mathbf{x}), \, \mathbf{x} \in \mathcal{S}$. As the power flow equation set $ h(\cdot)$ are nonlinear, their inverse functions are not easy to characterize. Further, this inverse non-linearity is a major bottleneck in solving PLF using parametric methods.
The non-linearity of \eqref{eq:pf} mainly restrict efficient decision making under uncertainty as we cannot obtain the solution space tractably. More importantly, any individual state variable gets affected by the complete set of nonlinear equations as multiple voltages are present in \eqref{eq:pf} with $V_l$. This non-linearity directly depends on the network graph as the values of $G_{ml}, B_{ml}$ are zero when there is no connection between node $m$ and $l$. In such a case, $V_m$ will not directly appear in \eqref{eq:pf}. Thus, for different networks, the non-linearity affects various voltages differently. Therefore, the uncertainty propagation requires various approximations and complex formulations \cite{fan2012probabilistic,amid2018cumulant} for a specific type of network and uncertainty distribution. Therefore, these methods cannot be generalized directly. 

The dependence of various PLF methods on uncertainty distribution information is also a significant difficulty. The PV generation and EV penetration uncertainty distributions do not follow specific distributions. There is a big challenge in estimating the distribution associated with net power injection uncertainty, uniquely when multiple uncertain loads (EVs) and generation sources (PVs) are integrated into the system. To deal with these issues of non-linearity and lack of uncertainty information, we present the main result of GP based NP-PLF below.

\subsection{Main Result}

Following the analogy presented in \eqref{eq:reg}, the Newton-Raphson load flow (NRLF) solution obtained for input sample $\mathbf{x}^{i}$ can be interpreted as $y_l^{i}$ containing $f_l(\mathbf{x}^{i})$ for bus $l$ with numerical computation noise as:
 \begin{align}\label{eq:yhat}
      y_l^{i} = f_l(\mathbf{x}^{i}) + \varepsilon.
 \end{align}
Based upon \eqref{eq:yhat}, using the GP learning \eqref{eq:gp}, the posterior distribution parameters $\mu(\mathbf{x})$ and $\sigma(\mathbf{x})$ are obtained. Further, a straight forward method would be to keep obtaining more and more training samples $(f_l(\mathbf{x}^i), \mathbf{x}^i )$, and keep updating the posterior distribution parameters. This approach has some major difficulties. This method does not provide any bound on the possible uncertainty of posterior distribution. Using this method, we have no criteria to stop learning, and it has the same issue as MCS with the question: \textit{How many points are sufficient for learning? How much is the confidence in the output predictions?} 

\subsubsection{NP-PLF Learning}

To overcome these difficulties, we present the probabilistic learning bound (PLB). The PLB defines a range within which the target state variable $y_l$ will remain with the given probability. 
% This interpretation allows us to apply GP to learn the state vector using Newton-Rapson load flow (NRLF).
% Yet, this general framework allows to use any type of method or even any real data for power flow solution set points. The method is directly extendable to other power flows such as unbalanced power flow for distribution system. 
Now, based on GP-UCB regret bound \cite{srinivas2012information}, we present the PLB for any general random state variable  $y_l=f_l(\mathbf{x})$ (e.g., $|V_l|=f_l(\mathbf{x})$ ) for bus $k$, in the uncertain state subspace $\mathbf{x} \in \mathcal{S}$. If one concerns only the voltage magnitude and the real power injection $\mathbf{P} = \{P_l\}$ where $P_l = \textrm{real}(S_l)$, let's learn function $|V_l|=f_l(\mathbf{P})$. However, the voltage-reactive power,  and voltage-apparent power function can be learnt in the same manner.

\begin{theorem}\label{theo:NPPLF}
For a given $\delta \in (0,1) $, for any uncertain power system input vector $\mathbf{P}$, the inverse power flow solution function $|V_l|(\mathbf{P})$ for bus $l$ will be bounded with probability $1-\delta$ as
\begin{align}\label{eq:bound}
    \mu_N(\mathbf{P})- \xi_{max} \leq |V_l|(\mathbf{P}) \leq \mu_N(\mathbf{P})+ \xi_{max}
\end{align}
where $\xi_{max}=\max_{\mathbf{P}\in \mathcal{P}} \big \{\sqrt{\beta_{N+1}} \, \sigma_N(\mathbf{P})\big \}$ is PLB, and $\mathbf{P}$ lies in the uncertain state subspace $\mathcal{P}$ which corresponds to the projection of $\mathcal{S}$ to the real space. $N$ indicate number of training samples.
\end{theorem}

\begin{proof}
The proof follows Theorem 6 in \cite{srinivas2012information} and Theorem 1 in \cite{zhai2019region}. The result on regret bound, after $N$ GP-UCB sampled training  points, provides the relation as \cite{srinivas2012information}: 
\begin{align}\label{eq:rbound}
     \bigg | |V_l|(\mathbf{P})-\mu_{N}(\mathbf{P})  \bigg | \leq \sqrt{\beta_{N+1}} \, \sigma_{N}(\mathbf{P}), \quad \forall \mathbf{P}\in \mathcal{P}.
\end{align}

Upon opening the modules, we obtain upper and lower bounds as a function of $\mathbf{P}$. Applying the maximum operator over the regret bound obtained after $N$ sampling iterations of GP-UCB, $\xi= \sqrt{\beta_{N+1}} \, \sigma_N(\mathbf{P})$, we obtain the PLB \eqref{eq:bound}. $\blacksquare$
\end{proof}

\begin{algorithm}[h]
	\caption{NP-PLF}
	\label{algo}
	\begin{algorithmic}
		\REQUIRE { $\mathcal{S}$ , $ \mathcal{X} \in \mathbb{R}^n$, $\delta$, $\xi_{max}$, $\mu_o$, $k(\mathbf{x},\mathbf{x}')$, $\mathbf{x}^{1}$, $ y_l(\mathbf{x}^{1})$ }
		\ENSURE {$\mathbf{x}^{i}$, $ y(\mathbf{x}^{i})$, $\mu_i(\mathbf{x})$, $\sigma_i(\mathbf{x})$, $\beta_{i+1}$; $i \in \{1, \dots , N\}$}
% 		\vskip 0.2em 
		\hrulefill
		\vskip 0.2em
		\WHILE {($ \max_{\mathbf{x} \in \mathcal{S}} \{\sqrt{\beta_{i+1}}\,\,\sigma_i(\mathbf{x})\}\leq \xi_{max} $)}
		\STATE {Select $\mathbf{x}^{i}= \argmax_{\mathbf{x} \in \mathcal{S}}  \mu_{i}(\mathbf{x}) + \sqrt{\beta_{i+1}}\,\,\sigma_{i}(\mathbf{x})$}
 		\IF {$\mathbf{x}^{(i)} \in \mathcal{X}$}
% 		\STATE {Sample $y(\mathbf{x}^{(i)})=f(\mathbf{x}^{(i)}) + \varepsilon $ from \eqref{eq:lf} by NRLF}
% 		\STATE {Update  $\mu_i(\mathbf{x})$ and  $\sigma_i(\mathbf{x})$ from \eqref{eq:gp}}
% 		\ELSE
		\STATE {Select randomly $\mathbf{x}^{i}$}
		\ENDIF
		\STATE {Sample $ y(\mathbf{x}^{i})=f(\mathbf{x}^{(i)}) + \varepsilon $  by NRLF}
		\STATE {Update  $\mu_i(\mathbf{x})$ and  $\sigma_i(\mathbf{x})$ from \eqref{eq:gp}}
% 		\STATE {Update $\mathcal{X}=\mathcal{X} \cup \mathbf{x}^{i}$}
		\STATE {Update $\beta_{i+1}=2\|f\|_l^2+300\gamma_N\ln^3{(i/\delta)}$}
		\STATE {Update $\mathcal{X}=\mathcal{X} \cup \mathbf{x}^{i}$ and $i=i+1$}
		\ENDWHILE 
	\end{algorithmic}
\end{algorithm}

Theorem \ref{theo:NPPLF} provides a probabilistic guarantee for bounds of power flow solution. Further, the  probabilistic learning bound $\xi_{max}$ can serve as convergence criteria for the NP-PLF algorithm. The NP-PLF algorithm, designed based on Theorem \ref{theo:NPPLF}, is given as Algorithm \ref{algo}. The $\mathcal{X}$ represents set of training points and similar to description of matrix $\mathbf{X}$ given in section \ref{sec:gp}. The  $\|f\|_l$ indicate reproducing kernel Hilbert space (RKHS) norm while $\gamma_N$ is constant \cite{zhai2019region}. The prior mean $\mu_o$, and covariance function (kernel) $k(\mathbf{x},\mathbf{x'})$  are selected based on our prior understanding about $f_l$. The $y_l(\mathbf{x}^1)$ represents output at initial random input vector $\mathbf{x}^1$. As mentioned earlier, the input $\mathbf{x}$  %The details of RKHS norm calculation can be found in appendix of \cite{zhai2019region}.

\begin{remark}
The proposed NP-PLF learning does not require any specific PDF of input variable $\mathbf{P}$. It can be seen as working with \textit{uniform probability distribution} where every value in rectangular region, $\mathcal{P}= \{\mathbf{P}|\mathbf{P}^{min}\leq \mathbf{P} \leq \mathbf{P}^{max}\}$, has same probability of occurrence. 
\end{remark}

\subsubsection{NP-PLF Testing }
In this step, once the Algorithm \ref{algo} converges, the state variable $y_l$ can be obtain for distribution of test points using $\mathbf{P}'$ in \eqref{eq:gp}. This means that for obtaining a probabilistic distribution of state variable against any input PDF, we use input PDF as test points over the output of the Algorithm \ref{algo}, $\mu_N(\mathbf{P})$ and $\sigma_N(\mathbf{P})$. Further, from Theorem \ref{theo:NPPLF}, the state PDF will be bounded by error $\xi_{max}$ with probability $1-\delta$. Therefore, the testing step can be performed for any class of probability distribution.

%The presented NP-PLF can also be looked from the angle of continuation power flow (CPF) \cite{ajjarapu1992continuation}. 
Here, it is important to highlight the difference between the proposed NP-PLF and the continuation power flow (CPF) method. The CPF attempts to trace the $P-V$ curve by increasing the load, generally in one dimension \cite{ajjarapu1992continuation}. The proposed method can learn the voltage variation curve in multiple dimensions as $\mathbf{x} \in \mathbb{R}^{2n}$ for $n$-bus network. The input dimension taken as $2n$ because complex power injection ($S_l$) is separated as real ($P_l$) and reactive power ($Q_l$) injection inputs ($S_l=P_l+jQ_l$) to capture their effects separately. Therefore, Algorithm 1 can also be interpreted as a GP based CPF method for $2n-dimensional$ input space. Nevertheless, it still suffers from the same issue of Jacobian ill-conditioning near the critical point \cite{iba1991calculation,jiang2019boundary, hungunsolvable, wolter2019differential}. Eliminating this assumption and finding the critical loading point is an objective of our ongoing research.  %Yet, it is an excellent tool for learning voltage curve in $n-dimensional$ input space.  

More importantly, Algorithm \ref{algo} %indicate learning one state variable $y$ at a time. Yet, this algorithm 
can work in parallel to learn many state variables simultaneously. Each such effort will sample points to learn the particular $y_l$ using GP-UCB \eqref{eq:gpucb} in parallel execution. %As the proposed method learn solution of entire set of equation $f(\mathbf{x})$ together, each parallel execution will be independent form each other and thus facilitate fast computation. 
Nevertheless, Multi-linear Gaussian Processes (MLGP) \cite{yu2017tensor}, AutoGP \cite{krauth2016autogp} can also be used to learn several variables simultaneously with Algorithm \ref{algo}. 

% \begin{figure}[t]
%     \centering
%     \includegraphics[width=0.8\columnwidth]{figure/EEE-30-bus-test-system.png}
%     \vspace{-1em}
% 	\caption{ Single line diagram of IEEE 30-Bus system \cite{li2010investigation}.}
%     \label{fig:30bus}
%     \vspace{-0.5em}
% \end{figure}

\section{Results and Discussion}

At first, we consider a situation when the renewable generator replaces the conventional generator into the IEEE 30-Bus system \cite{zimmerman2010matpower} %shown in fig. \ref{fig:30bus}. 
Later, the load uncertainties are considered in 30-Bus and 118-Bus systems \cite{zimmerman2010matpower} at different buses. The biggest challenge in PLF involves the difficulty in modeling and obtaining the statistical distribution of random power generation variable, more so in case of solar power. Therefore, for the NP-PLF learning stage, we consider uniform distribution between zero to the maximum limit of renewable generator. However, other distributions can be considered in a similar way. In the table \ref{tab:err1}, number of training samples ($N$) required, to achieve $\xi_{max} \leq 1\%$ with probability $\geq 0.99$, is given with computation time. It is important to note here that as we are not using any specific distribution of uncertainty, we need not define and fix the type of renewable generator and NP-PLF can work with any type of source.

\begin{table}[t]
  \centering
  \caption{Average Percentage Relative Error in NP-PLF Solution Compared to 50000 MCS Simulations}
  \bgroup
\def\arraystretch{1}
    \begin{tabular}{c|cccc}
      \multirow{2}{*}{\shortstack{Uncertain Renewable \\ Generator}} & Variable & $\% \, \varepsilon_v$ & $N$ & \multirow{2}{*}{\shortstack{Time\\ (sec.)}} \\ &  &  &  \\
          \hline
     \multirow{2}{*}{\shortstack{$P_{g_3}$  \\ At bus $22$ }} & $|V_{21}|$    & 0.0014
 & 8 & 2.15 \\
          & $|V_{24}|$    & 0.0001 & 7 & 1.39\\
    \hline
         \multirow{2}{*}{\shortstack{$P_{g_4}$  \\ At bus $27$ }} & $|V_{25}|$    & 0.0082

 & 7 & 1.66 \\
          & $|V_{28}|$    & 0.0055 & 8 & 1.88\\
    \hline
     \multirow{2}{*}{\shortstack{$P_{g_5}$  \\ At bus $13$ }} & $|V_{15}|$    & 0.0006
 & 13 & 2.92 \\
          & $|V_{24}|$    & 0.0049 & 7 & 1.89\\
    \hline
       \end{tabular}%
    \egroup
  \label{tab:err1}%
\end{table}%

\begin{figure}[t]
    \centering
    \includegraphics[width=0.9\columnwidth]{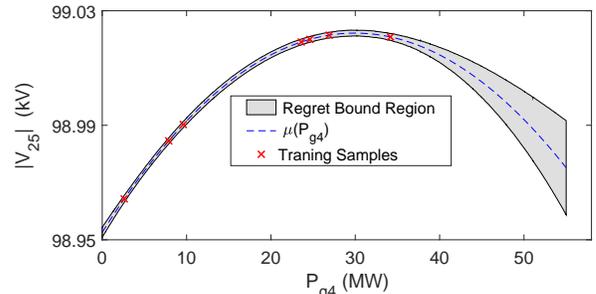}
    \vspace{-1.0em}
	\caption{ $|V_{25}|$ as a function of uncertain renewable power injection ($P_{g_4}$) at node 27 having uniform distribution between zero to 55 MW.}
    \label{fig:rg4_V25_gp}
\end{figure}

It is clear from table \ref{tab:err1} that for various cases, the proposed NP-PLF has been able to achieve the higher accuracy results in very less time when compared to the MCS method. Further, the fig.\ref{fig:rg4_V25_gp} shows the voltage variation obtained as ``\textit{semi-explicit}'' form via learning Algorithm \ref{algo} for one dimensional input subspace where $P_{g_4}$ indicate random real power generator by generator 4 connected at 27-th bus. The curve also indicates that learning regret is higher at locations where training samples are not obtained. Thus, regret bound region can further be decreased using more samples, especially with a higher value of $P_{g_4}$. Most importantly, upon completion of the NP-PLF learning stage, the complete $P-V$ curve is obtained. The average percentage relative error index $\% \, \varepsilon_v$, for $N_s$ testing samples, is defined as \cite{nosratabadi2018nonparametric}:
\begin{align}
    \% \, \varepsilon_v =\frac{\sum_{k=1}^{N_s}\bigg |\dfrac{V^{MCS}_l-V^{GP}_l}{V^{MCS}_l}\bigg |}{N_s} \times 100
\end{align}

% \begin{figure}[t]
%     \centering
%     \includegraphics[width=\columnwidth]{figure/rg3_V21_gp.eps}
%     \vspace{-0.8em}
% 	\caption{$|V_{21}|$ as a function of uncertain renewable power injection ($P_{g_3}$) at node 22}
%     \label{fig:rg3_V21_gp}
%           \vspace{-1em}
% \end{figure}

As mentioned before, the proposed NP-PLF method can work with any class of input uncertainty distribution through testing phase. The fig. \ref{fig:rPg4_normal_GP_MCS} and fig. \ref{fig:rPg4_gamma_GP_MCS} are obtained by varying the $P_{g_4}$ with normal and gamma distribution respectively. The comparison between histograms obtained using MCS and proposed NP-PLF method validates that proposed method can calculate statistical features, mean and standard deviation, of PLF output for any type of input uncertainty distribution.  Here, in fig. \ref{fig:rPg4_normal_GP_MCS}, the distribution is one sided because the $|V_{25}|_{max}=99.02 kV$ as indicative in the fig.\ref{fig:rg4_V25_gp}. The maximum number of samples in $P_{g_4}$ is around the mean (28 MW) which leads to voltage near maximum value. Thus one sided distribution of $|V_{25}|_{max}$ is obtained in fig.\ref{fig:rPg4_normal_GP_MCS}. 

\begin{figure}[t]
    \centering
    \includegraphics[width=\columnwidth]{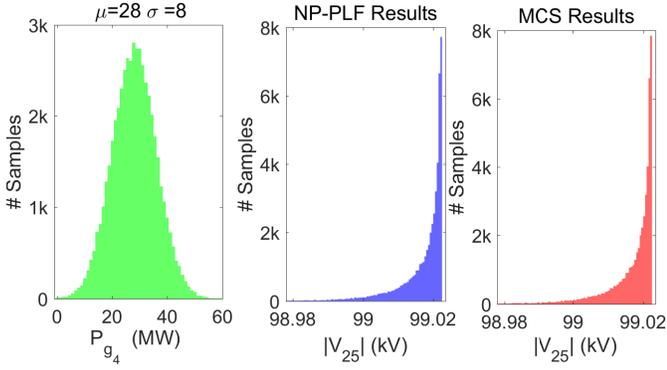}
    \vspace{-1.5em}
	\caption{Histogram for $|V_{25}|$ with uncertain $P_{g_4}$ at node 27 having normal distribution using proposed NP-PLF and MCS method (50000 samples). Error result is in Table \ref{tab:err1}.}
    \label{fig:rPg4_normal_GP_MCS}
          \vspace{-0.1em}
\end{figure}

\begin{figure}[t]
    \centering
    \includegraphics[width=\columnwidth]{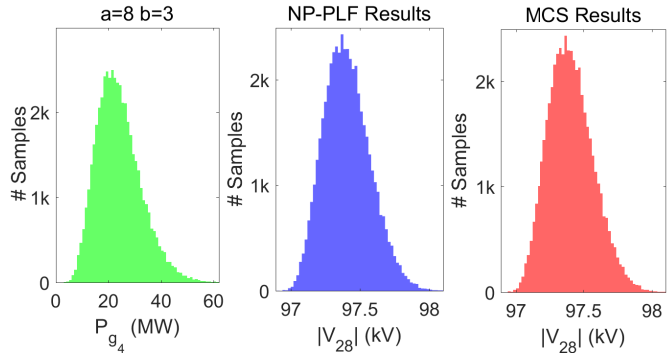}
    \vspace{-1.5em}
	\caption{Histogram for $|V_{28}|$ with uncertain $P_{g_4}$ at node 27 having gamma distribution (shape parameter a=8 and scale parameter b=3) using proposed NP-PLF and MCS method (50000 samples). Error result is in Table \ref{tab:err1}}
    \label{fig:rPg4_gamma_GP_MCS}
           \vspace{-0.0em}
\end{figure}

\begin{table}[b]
  \centering
%   \caption{$\xi_{max}$, $N$ and Computation Time Corresponding to NP-PLF Solution Shown in fig.\ref{fig:rPd30_Voltages_30bus}}
   \caption{Results Corresponding to NP-PLF Solution Shown in fig.\ref{fig:rPd30_Voltages_30bus}}
  \bgroup
\def\arraystretch{1}
    \begin{tabular}{c|ccc}
       Variable & $ \xi_{max} (kV)$ & $N$ & Time (sec.) \\
          \hline
         $|V_{30}|$    & 0.0784 & 10 & 2.35 \\
    % \hline
          $|V_{29}|$    & 0.0308 & 15 & 3.54 \\
    % \hline
     $|V_{28}|$    & 0.0443  & 8 & 1.89 \\
      $|V_{26}|$    & 0.0200 & 8 & 1.90 \\
    \hline
       \end{tabular}%
    \egroup
  \label{tab:er_timenode30}%
\end{table}%

For indicating the effect of load variation, different node voltage magnitudes are expressed in ``\textit{semi-explicit}'' form using the learning step and shown in fig. \ref{fig:rPd30_Voltages_30bus}. As indicated, the $|V_{30}|$ gets affected maximum with variations in load demand at node 30 while as we move away, the effect decreases largely. Yet, it is clear that the proposed method has been able to record the complete non-linearity of power flow equations and effect on all node voltages. Table \ref{tab:er_timenode30} shows that low values of $\xi_{max}$ are obtained with a very less number of training samples. These results show the higher accuracy and speed of the proposed method. The fig. \ref{fig:rPQd_V75_118bus} shows $|V_{75}|$ variation in $2-dimensional$, $P_{d_{75}}-Q_{d_{75}}$ space. It is important to understand that for this learning, MCS would require very large number of points while proposed method has been able to do this with very less points. The fig.\ref{fig:rPQd_V75_118bus} is drown with $\xi_{max} \leq 1\%$ and probability $\geq 0.99$.

% \begin{figure}[t]
%     \centering
%     \includegraphics[width=\columnwidth]{figure/rg5_V24_gp.eps}
%     \vspace{-0.8em}
% 	\caption{$|V_{24}|$ as a function of uncertain renewable power injection ($P_{g_5}$) at node 13}
%     \label{fig:rg5_V24_gp}
%           \vspace{-1em}
% \end{figure}

% \begin{figure}[h]
%     \centering
%     \includegraphics[width=\columnwidth]{figure/rPd17_theta_30bus.eps}
%     \vspace{-0.8em}
% 	\caption{Node voltage angle as a function uncertain power demand at node 17 in 30-Bus system}
%     \label{fig:rPd17_theta_30bus}
%           \vspace{-1em}
% \end{figure}
% \begin{table}[h]
%   \centering
%   \caption{$\xi_{max}$, $N$ and Computation Time Corresponding to NP-PLF Solution Shown in fig.\ref{fig:rPd17_theta_30bus}}
%   \bgroup
% \def\arraystretch{1.2}
%     \begin{tabular}{c|ccc}
%       Variable & $ \xi_{max} (rad)$ & $N$ & Time (sec.) \\
%           \hline
%          $\theta_{17}$    & 0.00036 & 5 & $\leq 1$ \\
%     % \hline
%           $\theta_{16}$    & 0.00095 & 5 & $\leq 1$ \\
%     % \hline
%      $\theta_{10}$    & 0.00027  & 5 & 1.91 \\
%       $\theta_{9}$    & 0.03017 & 9 & 2.18 \\
%     \hline
%       \end{tabular}%
%     \egroup
%   \label{tab:er_timenode17}%
% \end{table}%

In the following, we report the time consumption for the testing stage with unoptimized codes. The GPML toolbox \cite{rasmussen2010gaussian} with MATLAB 2018b on PC having Intel Xeon E5-1630v4@3.70 GHz, 16.0 GB RAM is used for simulations. The time consumption increases with increment in the input subspace size. For larger problems, works on approximation methods (chapter 8 \cite{williams2006gaussian}), \cite{krauth2016autogp} can be used for future works. Also, the proposed NP-PLF is divided into stages where the learning stage can be done offline, improving the overall computational performance. Testing is very less time consuming, and 50000 points take 0.073563 seconds only to test in fig.\ref{fig:rPg4_gamma_GP_MCS} while 50000 samples take 108.15 seconds with MCS.

\section{Conclusion}
In this paper, a novel non-parametric probabilistic load flow (NP-PLF) is presented to estimate the nodal voltages for uncertain power injections. The proposed method attempts to understand the power balance equality constraint under uncertain input space to improve the power system's decision-making process. The proposed NP-PLF consists of two steps. The learning step has been built on GP regression, and voltage solution has been learned as a function of arbitrary power injections providing a ``\textit{semi-explicit}'' form with probabilistic learning bound (PLB). Then, the testing step has shown to approach the final voltage distribution, in terms of inverse power flow solutions, for any class of power injection uncertainty distribution. The proposed algorithm was tested in the IEEE 30-bus and IEEE 118-bus systems. The simulation results prove that the NP-PLF method can obtain statistical information using very few sampling points. Also, the relative error-index is sufficiently small, while the computational time is much less when compared with the traditional MCS method. Future works involve the development of multiple applications based on the ``\textit{semi-explicit}'' learning method developed in this work. 

\begin{figure}[t]
    \centering
    \includegraphics[width=\columnwidth]{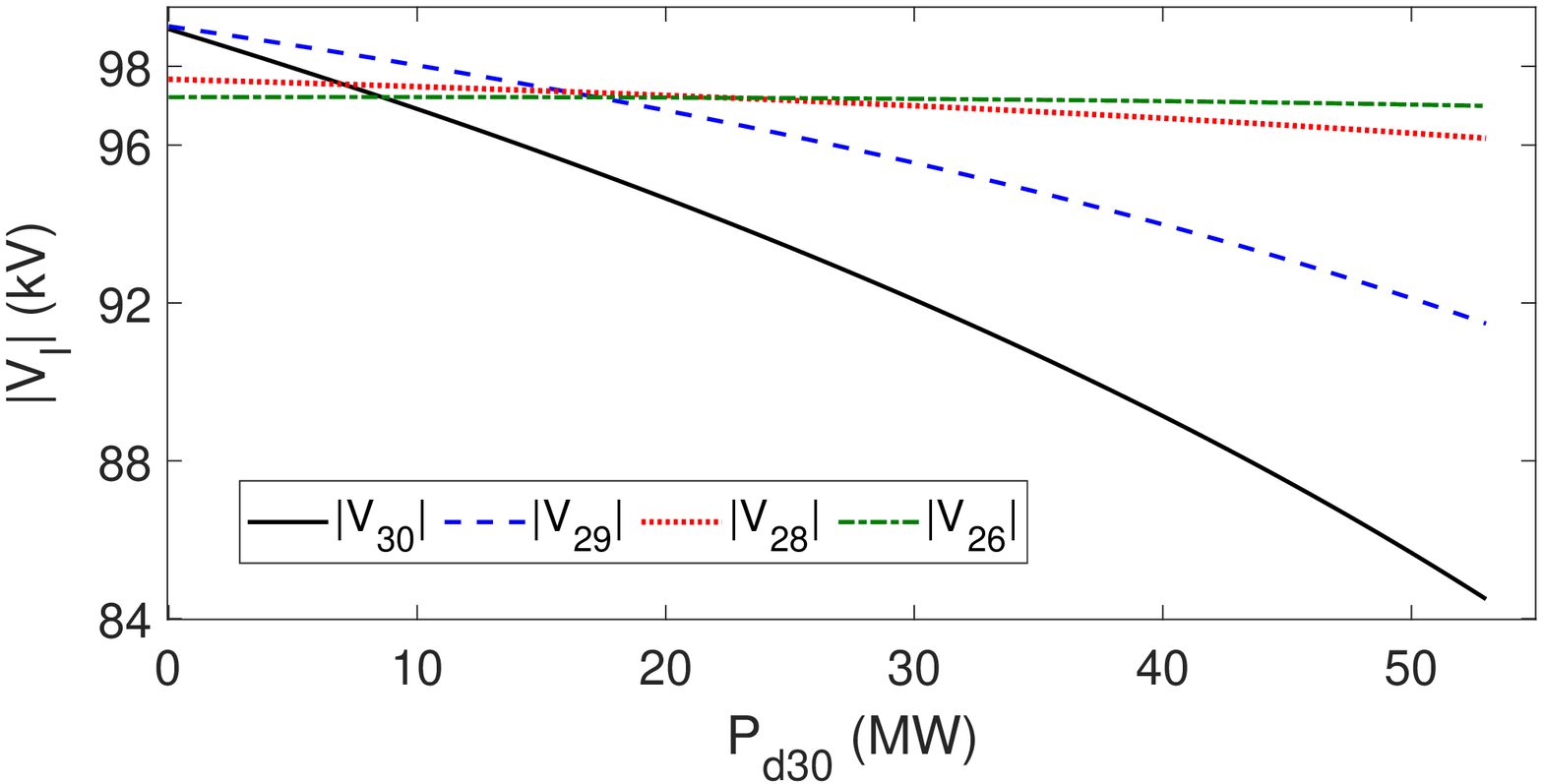}
    \vspace{-2.0em}
	\caption{$|V_l|$ as a function of uncertain load $P_{d_{30}}$ in 30-Bus system}
    \label{fig:rPd30_Voltages_30bus}
         \vspace{-0.1em}
\end{figure}

\begin{figure}[t]
    \centering
    \includegraphics[width=\columnwidth]{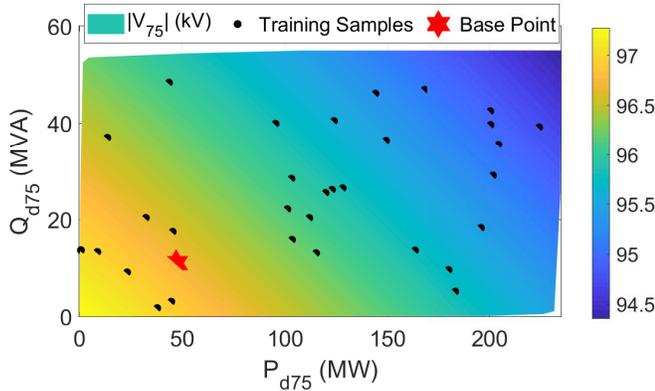}
    \vspace{-2.5 em}
	\caption{$|V_{75}|$ as a function of uncertain load $S_{d_{30}}$ in 118-Bus system}
    \label{fig:rPQd_V75_118bus}
         %\vspace{-1em}
\end{figure}

\section*{Acknowledgement}
Parikshit Pareek, Wang Chuan, and Hung D. Nguyen are supported by NTU SUG, MOE, EMA and NRF fundings.

\bibliographystyle{IEEEtran}
\bibliography{main}

\begin{thebibliography}{10}
\providecommand{\url}[1]{#1}
\csname url@samestyle\endcsname
\providecommand{\newblock}{\relax}
\providecommand{\bibinfo}[2]{#2}
\providecommand{\BIBentrySTDinterwordspacing}{\spaceskip=0pt\relax}
\providecommand{\BIBentryALTinterwordstretchfactor}{4}
\providecommand{\BIBentryALTinterwordspacing}{\spaceskip=\fontdimen2\font plus
\BIBentryALTinterwordstretchfactor\fontdimen3\font minus
  \fontdimen4\font\relax}
\providecommand{\BIBforeignlanguage}[2]{{%
\expandafter\ifx\csname l@#1\endcsname\relax
\typeout{** WARNING: IEEEtran.bst: No hyphenation pattern has been}%
\typeout{** loaded for the language `#1'. Using the pattern for}%
\typeout{** the default language instead.}%
\else
\language=\csname l@#1\endcsname
\fi
#2}}
\providecommand{\BIBdecl}{\relax}
\BIBdecl

\bibitem{da2018risk}
A.~M.~L. da~Silva and A.~M. de~Castro, ``Risk assessment in probabilistic load
  flow via monte carlo simulation and cross-entropy method,'' \emph{IEEE Trans.
  on Power Systems}, vol.~34, no.~2, pp. 1193--1202, 2018.

\bibitem{borkowska1974probabilistic}
B.~Borkowska, ``Probabilistic load flow,'' \emph{IEEE Trans. on Power Apparatus
  and Systems}, no.~3, pp. 752--759, 1974.

\bibitem{allan1981evaluation}
R.~Allan, A.~L. Da~Silva, and R.~Burchett, ``Evaluation methods and accuracy in
  probabilistic load flow solutions,'' \emph{IEEE Trans. on Power Apparatus and
  Systems}, no.~5, pp. 2539--2546, 1981.

\bibitem{hungzheng}
Z.~Zhang, H.~D. Nguyen, K.~Turitsyn, and L.~Daniel, ``Probabilistic power flow
  computation via low-rank and sparse tensor recovery,'' \emph{arXiv preprint
  arXiv:1508.02489}, 2015.

\bibitem{constante2018data}
G.~E. Constante-Flores and M.~S. Illindala, ``Data-driven probabilistic power
  flow analysis for a distribution system with renewable energy sources using
  monte carlo simulation,'' \emph{IEEE Trans. on Industry Applications},
  vol.~55, no.~1, pp. 174--181, 2018.

\bibitem{hatziargyriou1993probabilistic}
N.~Hatziargyriou, T.~Karakatsanis, and M.~Papadopoulos, ``Probabilistic load
  flow in distribution systems containing dispersed wind power generation,''
  \emph{IEEE Trans. on Power Systems}, vol.~8, no.~1, pp. 159--165, 1993.

\bibitem{nosratabadi2018nonparametric}
H.~Nosratabadi, M.~Mohammadi, and A.~Kargarian, ``Nonparametric probabilistic
  unbalanced power flow with adaptive kernel density estimator,'' \emph{IEEE
  Trans. on Smart Grid}, vol.~10, no.~3, pp. 3292--3300, 2018.

\bibitem{morales2007point}
J.~M. Morales and J.~Perez-Ruiz, ``Point estimate schemes to solve the
  probabilistic power flow,'' \emph{IEEE Trans. on power systems}, vol.~22,
  no.~4, pp. 1594--1601, 2007.

\bibitem{liu2018probabilistic}
C.~Liu, K.~Sun, B.~Wang, and W.~Ju, ``Probabilistic power flow analysis using
  multidimensional holomorphic embedding and generalized cumulants,''
  \emph{IEEE Trans. on Power Systems}, vol.~33, no.~6, pp. 7132--7142, 2018.

\bibitem{zhang2004probabilistic}
P.~Zhang and S.~T. Lee, ``Probabilistic load flow computation using the method
  of combined cumulants and gram-charlier expansion,'' \emph{IEEE Trans. on
  power systems}, vol.~19, no.~1, pp. 676--682, 2004.

\bibitem{ajjarapu1992continuation}
V.~Ajjarapu and C.~Christy, ``The continuation power flow: a tool for steady
  state voltage stability analysis,'' \emph{IEEE Trans. on Power Systems},
  vol.~7, no.~1, pp. 416--423, 1992.

\bibitem{williams2006gaussian}
C.~K. Williams and C.~E. Rasmussen, \emph{Gaussian processes for machine
  learning}.\hskip 1em plus 0.5em minus 0.4em\relax MIT press Cambridge, MA,
  2006, vol.~2, no.~3.

\bibitem{lee2013short}
D.~Lee and R.~Baldick, ``Short-term wind power ensemble prediction based on
  gaussian processes and neural networks,'' \emph{IEEE Trans. on Smart Grid},
  vol.~5, no.~1, pp. 501--510, 2013.

\bibitem{yan2015hybrid}
J.~Yan, K.~Li, E.-W. Bai, J.~Deng, and A.~M. Foley, ``Hybrid probabilistic wind
  power forecasting using temporally local gaussian process,'' \emph{IEEE
  Trans. on Sustainable Energy}, vol.~7, no.~1, pp. 87--95, 2015.

\bibitem{sheng2017short}
H.~Sheng and et. al., ``Short-term solar power forecasting based on weighted
  gaussian process regression,'' \emph{IEEE Trans. on Industrial Electronics},
  vol.~65, no.~1, pp. 300--308, 2017.

\bibitem{van2018probabilistic}
D.~W. van~der Meer and et. al, ``Probabilistic forecasting of electricity
  consumption, photovoltaic power generation and net demand of an individual
  building using gaussian processes,'' \emph{Applied energy}, vol. 213, pp.
  195--207, 2018.

\bibitem{zhai2019region}
C.~Zhai and H.~D. Nguyen, ``Region of attraction for power systems using
  gaussian process and converse lyapunov function--part i: Theoretical
  framework and off-line study,'' \emph{arXiv arXiv:1906.03590}, 2019.

\bibitem{pareek2019probabilistic}
P.~Pareek and H.~D. Nguyen, ``Probabilistic robust small-signal stability
  framework using gaussian process learning,'' \emph{21st Power Systems
  Computation Conference (PSCC)}, 2020, Accepted.

\bibitem{srinivas2012information}
N.~Srinivas, A.~Krause, S.~M. Kakade, and M.~W. Seeger, ``Information-theoretic
  regret bounds for gaussian process optimization in the bandit setting,''
  \emph{IEEE Trans. on Information Theory}, vol.~58, no.~5, pp. 3250--3265,
  2012.

\bibitem{fan2012probabilistic}
M.~Fan, V.~Vittal, G.~T. Heydt, and R.~Ayyanar, ``Probabilistic power flow
  studies for transmission systems with photovoltaic generation using
  cumulants,'' \emph{IEEE Trans. on Power Systems}, vol.~27, no.~4, pp.
  2251--2261, 2012.

\bibitem{amid2018cumulant}
P.~Amid and C.~Crawford, ``A cumulant-tensor-based probabilistic load flow
  method,'' \emph{IEEE Trans. on Power Systems}, vol.~33, no.~5, pp.
  5648--5656, 2018.

\bibitem{iba1991calculation}
K.~Iba, H.~Suzuki, M.~Egawa, and T.~Watanabe, ``Calculation of critical loading
  condition with nose curve using homotopy continuation method,'' \emph{IEEE
  Trans. on Power Systems}, vol.~6, pp. 584--593, 1991.

\bibitem{jiang2019boundary}
T.~Jiang, K.~Wan, and Z.~Feng, ``Boundary-derivative direct method for
  computing saddle node bifurcation points in voltage stability analysis,''
  \emph{International Journal of Electrical Power \& Energy Systems}, vol. 112,
  pp. 199--208, 2019.

\bibitem{hungunsolvable}
H.~N. Dinh, M.~Y. Nguyen, and Y.~T. Yoon, ``A new approach for corrective and
  preventive control to unsolvable case in power networks having ders,''
  \emph{Journal of Electrical Engineering and Technology}, vol.~8, no.~3, pp.
  411--420, 2013.

\bibitem{wolter2019differential}
F.-E. Wolter and B.~Berger, ``Differential geometric foundations for power flow
  computations,'' \emph{arXiv preprint arXiv:1903.11131}, 2019.

\bibitem{yu2017tensor}
R.~Yu, G.~Li, and Y.~Liu, ``Tensor regression meets gaussian processes,''
  \emph{arXiv preprint arXiv:1710.11345}, 2017.

\bibitem{krauth2016autogp}
K.~Krauth, E.~V. Bonilla, K.~Cutajar, and M.~Filippone, ``Autogp: Exploring the
  capabilities and limitations of gaussian process models,'' \emph{arXiv
  preprint arXiv:1610.05392}, 2016.

\bibitem{zimmerman2010matpower}
R.~D. Zimmerman, C.~E. Murillo-S{\'a}nchez, and R.~J. Thomas, ``Matpower:
  Steady-state operations, planning, and analysis tools for power systems
  research and education,'' \emph{IEEE Trans. on power systems}, vol.~26,
  no.~1, pp. 12--19, 2010.

\bibitem{rasmussen2010gaussian}
C.~E. Rasmussen and H.~Nickisch, ``Gaussian processes for machine learning
  (gpml) toolbox,'' \emph{Journal of machine learning research}, vol.~11, no.
  Nov, pp. 3011--3015, 2010.

\end{thebibliography}

\end{document}